\documentclass{amsart}
\usepackage{float}
\usepackage[pagewise]{lineno}\linenumbers
\usepackage[utf8]{inputenc}
\usepackage{amsthm,amsmath,amsfonts}
\usepackage{soul,color}
\usepackage{longtable}
\usepackage{xcolor}
\usepackage{tabularx}
\usepackage{longtable}
\usepackage{booktabs}
\usepackage{mathtools}
\usepackage{mathabx}
\newcommand{\Fq}{\mathbb{F}_q}
\newtheorem{theorem}{Theorem}[section]

\newtheorem*{remark}{Remark}

\usepackage{color,soul,xcolor}
\usepackage{multicol}
\usepackage{enumitem}
\usepackage{hyperref}
\usepackage{comment}
\usepackage{titlesec}
\theoremstyle{definition}
\newtheorem{definition}{Definition}[section]

\makeatletter
\def\section{\@startsection{section}{1}%
  \z@{.7\linespacing\@plus\linespacing}{.5\linespacing}%
  {\normalfont\scshape\centering}}
\def\subsection{\@startsection{subsection}{2}%
  \z@{.5\linespacing\@plus.7\linespacing}{-.5em}%
  {\normalfont\bfseries}}
\makeatother

\title{A Generalization of the ASR Search Algorithm to 2-Generator Quasi-Twisted Codes}
\author{Dev Akre, Nuh Aydin, Matthew J. Harrington, Saurav R. Pandey}
\date{}
\begin{document}
\nolinenumbers
\maketitle
\begin{abstract}
One of the main goals of coding theory is to construct codes with best possible parameters and properties. A special class of codes called quasi-twisted (QT) codes is well-known to produce codes with good parameters. Most of the work on QT codes has been over the 1-generator case. In this work, we focus on 2-generator QT codes and generalize the ASR algorithm that has been very effective to produce new linear codes from 1-generator QT codes. Moreover, we also generalize a recent algorithm to test equivalence of cyclic codes to constacyclic codes. This algorithm makes the ASR search even more effective. As a result of implementing our algorithm,  we have found 103 QT codes that are new among the class of QT codes. Additionally, most of these codes possess the following additional properties: a) they have the same parameters as best known linear codes, and b) many of the have additional desired properties such as being LCD and dual-containing. Further, we have also found a binary 2-generator QT code that is new (record breaking) among all binary linear codes  \cite{database} and its extension yields another record breaking binary linear code.  
\end{abstract}


\section{Introduction and Motivation}
Constructing codes with best possible parameters is one of the central and challenging problems in coding theory. Every linear code has three fundamental parameters: the length ($n$), the dimension ($k$), and the minimum (Hamming) distance $d$. Such a code over the finite field $\mathbb{F}_q$ is referred to as an $[n,k,d]_q$-code. 

\begin{comment}
A special class of codes called linear codes are much easier to work with due to their properties. Mathematically, a linear code $C$ of length $n$ over $\mathbb{F}_q$, the finite field with $q$ elements, is a vector subspace of $\mathbb{F}^n_q$. A linear code with length $n$, dimension $k$, and minimum distance $d$ is notated as $[n,k,d]_q$ code. For a given code, the number of errors it can correct is given by ($\left \lfloor \frac{d-1}{2} \right \rfloor $), where $d$ is the minimum distance of the code. By lengthening the linear code, we can improve its error-correcting capability but this brings about an increase in memory consumption and a decrease in the information rate $\left(\frac{k}{n}\right)$ of the code.
\end{comment}

Thus, we have an optimization problem where we want to determine the optimal value of the third parameter of a linear code given the values of the other two parameters. For  instance,  fixing $n$ and $k$, we look for the largest possible value of $d$, denoted $d_q[n,k]$. The online database (\cite{database}) gives known information about $d_q[n,k]$ for $q\leq 9$ and for code lengths up to 256 (different upper bound for $n$ for each field).  

%\hl{Talk about how comp infeasible to do comprehensive searches over arbitrary linear codes .. focus on promising types of codes...QT codes have been extremely fruitful .... - DONE}
However, this optimization problem is computationally taxing even with the help of modern computers. First, computing the minimum distance of an arbitrary linear code is NP-hard, and it becomes infeasible for large dimensions \cite{NPhard}. Second, for a given length, dimension, and finite field $\Fq$, the number of linear codes is large and grows very quickly. For these reasons, it becomes computationally infeasible to conduct comprehensive searches over arbitrary linear codes. As can be observed from the database \cite{database}, it most cases optimal codes are not known. In general, this optimization problem is solved when either $k$ or $n-k$ is relatively small.  

 A promising class of codes called quasi-twisted (QT) codes have been extremely fruitful in producing codes with good parameters. The search algorithm ASR, first introduced in \cite{ASR} for 1-generator QT codes, has been particularly effective. The ASR algorithm and its generalization \cite{GenASR} have been used in many subsequent works (e.g., \cite{qt1,qt2,qt3,qt4}) and produced many record breaking codes.  Most of the research in the literature on QT codes has been on the special case of 1-generator codes with only a few exceptions on multi-generator codes (e.g., \cite{2gen1,2gen2}). The main goal of this work is to generalize the ASR algorithm to 2-generator QT codes and test its effectiveness.

\section{ Basic Definitions and Preliminaries}
Cyclic codes have a prominent place in coding theory for both theoretical and practical reasons. Some of the most important classes of  codes such as binary Hamming Codes, BCH codes, Reed-Solomon codes, and quadratic residue codes are either cyclic or equivalent to cyclic codes. Theoretically, they establish a key link between coding theory and algebra via the correspondence that maps an arbitrary vector $(c_0,c_1,\dots,c_{n-1})$ in $\Fq^n$ to the polynomial $c(x)=c_0+c_1x+\cdots+c_{n-1}x^{n-1}$ of degree less than $n$. Based on this vector space isomorphism, we use vectors and polynomials interchangeably. 

\begin{definition}
A linear code $C$ is called cyclic if it is closed under the cyclic shift operator $\pi$, i.e., whenever $c=(c_0,c_1,...,c_{n-1})$ is a codeword of C, then so is 
$\pi(c) = (c_{n-1},c_0,.....,c_{n-2})$. 
\end{definition}

\noindent  Constacyclic  codes are generalizations of cyclic codes.

\begin{definition}
Let $a$ be a non-zero constant in $\mathbb{F}_q$. A linear code $C$ is called constacyclic (CC) if it is closed under the constacyclic shift operator, i.e., whenever $c=(c_0,c_1,...,c_{n-1})$ is a codeword of C, then so is 
$\pi_a(c) = (ac_{n-1},c_0,.....,c_{n-2})$.
\end{definition}
\noindent  It is well known that the CC shift corresponds to multiplication by $x \mod x^n-a$ and constacyclic codes are ideals in the quotient ring  $\mathbb{F}_q[x]/ \langle x^n-a \rangle$. Note that when the constant $a$ in the above definition, called the shift constant, is taken to be 1,  we obtain cyclic codes as a special case. Similarly to a cyclic code, a CC code $C$ of length $n$ over $\mathbb{F}_q$ with shift constant $a$ has a generator $g(x)$ that must divide $x^n-a$. Such a non-zero, monic polynomial of smallest degree  is unique and it is called the (standard) generator of $C$. A CC code $C$ is a principal ideal generated by $g(x)$, denoted $C=\langle g(x)\rangle$. Therefore, there is a one-to-one correspondence between CC codes of length $n$ with shift constant $a$ over $\mathbb{F}_q$ and the divisors of $x^n-a$. The dimension of $C=\langle g(x)\rangle$ is $k=n-\deg(g(x))$ with a basis $\{g(x),xg(x),\dots,x^{k-1}g(x)\}$. The polynomial $h(x)=(x^n-a)/g(x)$ is called the check polynomial for C. The check polynomial has the property that a word $v(x)$ is in $C$ if and only if $h(x)v(x)=0$ in $\mathbb{F}_q[x]/ \langle x^n-a \rangle$. A CC code $C=\langle g(x) \rangle$ has many generators and any generator is of the form $g(x)f(x)$ where $\gcd(f(x),h(x))=1$.  From the generator polynomial $g(x)=g_0+g_1x+\cdots+g_mx^m$ of a CC code $C$ we obtain its generator matrix as an $a$-circulant (twistulant) matrix
\begin{center}
\[ Circ(g)=
\left[
\begin{array}{ccccccc}
g_0&g_1&\cdots&g_{m}& 0 & \cdots &0\\
0&g_0&g_1&\cdots&g_{m}& 0 \cdots&0\\
\hdots \\
0& \hdots &0 &g_0& g_1& \cdots &g_{m} \\

\end{array}
\right]
\]
\end{center}
 
\noindent where each row is the cyclic (constacyclic) shift of the row above.

\subsection{ Quasi-Cyclic (QC) and Quasi-Twisted (QT) codes.}
 Quasi-cyclic (QC) and quasi-twisted (QT) codes are generalizations of cyclic and CC codes where the shift can occur by $\ell$ positions for any $\ell\geq 1$.  A linear code $C$ is said to be $\ell$-quasi-cyclic (QC) if for a positive integer $\ell$, whenever $c = (c_0,c_1,...,c_{n-1})$ is a codeword, $(c_{n-\ell},...,c_{n-1}, c_0, c_1,...c_{n-\ell-1})$ is also a codeword. Similarly, a linear code $C$ is said to be $\ell$-quasi-twisted (QT) if it has the property that for a positive integer $\ell$ that divides $n$, if $c = (c_0,c_1,...,c_{n-1})$ is a codeword, then so is $(ac_{n-\ell},...,ac_{n-1}, c_0, c_1,...c_{n-\ell-1})$. Such a code $C$ is called a QT code of index $\ell$, or an $\ell$-QT code. If $\gcd(\ell,n)=1$, then a QT code is CC so WLOG, we assume that $\ell$ is a divisor of $n$. Hence we assume $n=m\cdot \ell$. Algebraically a QT code of length $n=m\cdot \ell$ is an $R$-submodule of $R^{\ell}$ where $R=\mathbb{F}_q[x]/ \langle x^m-a \rangle$. A generator matrix of a QT code can be put into the form
\begin{center}
\[\begin{bmatrix}
    G_{1,1} & G_{1,2} & \dots & G_{1,\ell} \\
    G_{2,1} & G_{2,2} & \dots & G_{2,\ell} \\
    \vdots  & \vdots  & \ddots&\vdots \\
    G_{r,1} & G_{r,2} & \dots & G_{r,\ell}
\end{bmatrix}
\]
\end{center}
where each $G_{i,j}=Circ(g_{i,j})$ is a twistulant matrix defined by some polynomial $g_{i,j}(x)$.  Such a code is called an $r$-generator QT code. Most of the work on QT codes in the literature is focused on the 1-generator case. In this work, we consider 2-generator QT codes.

\section{2-Gen QT Codes}
 A generator of a 2-generator QT code of index $\ell$ has the following form
\[
\mathbf{g(x)}=
  \begin{bmatrix}
    g_{11}(x), g_{12}(x),  \cdots, g_{1\ell}(x)\\
    g_{21}(x), g_{22}(x), \cdots, g_{2\ell}(x)\\

  \end{bmatrix}
\]

with the corresponding generator matrix 

\[
G=
  \begin{bmatrix}
     Circ(g_{11}(x)) &  Circ(g_{12}(x)) & \cdots &  Circ(g_{1\ell}(x)) \\
     Circ(g_{21}(x)) &  Circ(g_{22}(x))& \cdots  &  Circ(g_{2\ell}(x))\\

  \end{bmatrix}.
\]

Our goal is to generalize the ASR algorithm that has been very effective to produce new linear codes, to the 2-generator case. The ASR algorithm was based on the following theorem. 

\begin{theorem}  \cite{ASR}
Let $C$ be a $1$-generator $\ell$-QT code over $F_q$ of length $n = m\ell$ and shift constant $a$ with the generator of the form:
$$
( f_1(x)g(x), f_2(x)g(x), . . . , f_{\ell}(x)g(x)),
$$
where $x^m -a  = g(x)h(x)$ and for all $i=1 ,..., \ell$, $gcd(h(x),f_i(x))=1$ . Then, $C$ is an $[n,k,d']_q$-code where $dim(C)=m-deg(g(x)),$ and $d' \geq \ell \cdot d$, where d is the minimum distance of the constacyclic code $C_g$ of length $m$ generated by $g(x)$.
\end{theorem}

In the original implementation of the ASR algorithm, the first step was to compute all CC codes of length $m$. When there are multiple codes of a given dimension $k$, a code with highest minimum distance was chosen, say with generator $g(x)$. Then a search for  QT codes is set up so that many generators of the form given in the Theorem above are considered for the same $g(x)$ but with different $f_i(x)$'s that satisfy the gcd condition. Later, the ASR algorithm was generalized in \cite{GenASR}.  In the more generalized version, all constacyclic codes for a given length and dimension are partitioned into equivalence classes and one generator polynomial from each equivalence class is used to set up a QT search. The general algorithm uses more generator polynomials $g(x)$, hence it is more comprehensive. In fact, it is shown in \cite{GenASR} that record breaking codes were obtained with the more general algorithm, that would have been missed by the original algorithm because they came from generators of CC codes that did not have highest minimum distances. 

Notion of equivalent codes was fundamental to the generalization in \cite{GenASR}.  Codes that are equivalent share the same parameters and weight distributions, so it is sufficient to pick only one code from each equivalence class. Two linear codes are called equivalent to each other if one is obtained from the other by using any combination  of the following operations:

\begin{enumerate}
    \item Permutation of the coordinates.
    \item Multiplication of elements in a fixed position by a non-zero scalar in $\mathbb{F}_q$.
    \item Applying an automorphism of $\mathbb{F}_q$ to each component of the vectors.
\end{enumerate}

When only the first transformation (permutation) is used, the resulting codes are said to be permutation equivalent, which is an important special case. In a recent work, an efficient algorithm to check equivalence of cyclic codes was given \cite{cycliceq}. More recently. it has been generalized to CC codes \cite{CCeq}. We made use of this algorithm in our search.

%\hl{Change $l$ to  $\ell$ everywhere - SHOULD BE GOOD NOW}

We begin the ASR search algorithm for 2-generator QT codes by taking a generator $g(x)$ of a CC code over $\mathbb{F}_q$ with length $m$. Then, we construct the generator of an $\ell$-QT code as in the 1-generator case of ASR.
$$( f_1(x)g(x), f_2(x)g(x), . . . , f_l(x)g(x)),$$
\noindent where all $f_i(x)$ are chosen arbitrarily from $\mathbb{F}_q[x]/ \langle x^m-1\rangle$ such that they are relatively prime to $h(x)$, the check polynomial of the CC code generated by $g(x)$, and $\deg(f_i(x))<\deg(h(x))$. The second block has a similar structure but we often take one component 0. The following theorem gives more specifics about the form of a generator for a 2-QT code we considered.
%We repeat the same process again to construct the generator for the second block, and combine it with that of the first block to form a generator matrix for our 2-generator QT code.

\begin{theorem}
Let $g$ be the standard generator of a CC code $C_g$ of length $m$ over $\mathbb{F}_q$, i.e., $x^m-a=gh_1$. Let $p$ be a polynomial over $\mathbb{F}_q$ such that $p$ divides $h_1$. Let $h_2 = \frac{(x^m - a)}{p \cdot g}$. Let $C$ be a 2-gen QT code with a generator  of the form 
$$ \mathbf{g(x)}=\begin{bmatrix}
    gf_{11} & gf_{12} &  \cdots & gf_{1\ell} \\
    0 & pgf_{22} &  \cdots & pgf_{2\ell}
  \end{bmatrix} $$

\noindent where $\gcd(f_{1j},h_1)=1$ and $\gcd(f_{2j},h_2)=1$. Let $k_1 = m - deg(g)$ and let $k_2 = m - deg(p \cdot g)$. Then, the dimension of the QT code generated by $\mathbf{g(x)}$ is $k = k_1 + k_2$ and $d(C) \geq d$ where $d$ is the minimum distance of the CC code generated by $g$. 
\end{theorem}
%\hl{stopped here ...}
\begin{proof}
Let \[
G=
  \begin{bmatrix}
     Circ(gf_{11}(x)) &  Circ(gf_{12}(x)) & \cdots &  Circ(gf_{1\ell}(x)) \\
     0 &  Circ(pgf_{22}(x))& \cdots  &  Circ(pgf_{2\ell}(x))\\

  \end{bmatrix}.
\] 
be a generator matrix of $C$ and consider $Circ(gf_{11}(x))$ on the first row of $G$ and the $k \times m$ zero matrix on the second row of $G$. Any codeword $c \in C$ is a linear combination of  the rows $\{r_1,r_2,\dots,r_{k_1+k_2}\}$ of $G$. Thus, $c = a_1 \cdot r_1 + a_2 \cdot r_2 + \ldots +  a_{k_1} \cdot r_{k_1} + b_1 \cdot r_{k_1+1} + \ldots + b_{k_2} \cdot r_{k_1 + k_2}$ for some $a_i\in \mathbb{F}_q$. Consider the first $m$ coordinates of $c$. As the first $m$ coordinates of $r_i$, $i=k_1+1,k_1+2,\ldots,k_1 + k_2$, are zeroes, they are determined by rows $r_j$, $j=1,2,\ldots,k_1$. When considering the first $m$ coordinates of a codeword, we only need to consider $gf_{11}$ and $Circ(gf_{11}(x))$.

Let $gf_{11} = p_0 + p_1x + \cdots + p_{m-1}x^{m-1}$ and let its twistulant matrix be \[
Circ(gf_{11}(x))=
  \begin{bmatrix}
    p_0 & p_{1} & p_{2} & \cdots & p_{m-1} \\
    ap_{m-1} & p_{0} & p_{1} & \cdots & p_{m-2} \\
    ap_{m-2} & ap_{m-1} & p_{0} & \cdots & p_{m-3} \\
    \vdots & \vdots &  \vdots &  &  \vdots  \\
    ap_{m-k+1} & ap_{m-k+2} & ap_{m-k+2} & \cdots & p_{m-k}
  \end{bmatrix}.
\]
Note that the $k_1$ rows of $Circ(gf_{11})$ are linearly independent. Now let $c= \vec{0} \in C$. In particular, the first $m$ coordinates of $c$ is the zero vector. Since the rows of $Circ(gf_{11})$ are linearly independent,  it follows that for all $j=1,2,\ldots,k_1$, $a_j = 0$. Thus, $c = \vec{0} =  b_1 \cdot r_{k_1+1} + \ldots + b_{k_2} \cdot r_{k_1 + k_2}$. Now, since the rows of the matrices  $Circ(pgf_{22}(x)), \cdots,  Circ(pgf_{2\ell}(x))$ are linearly independent, $b_i=0$  for all $i=1,2,\ldots,k_2$.  Hence the rows of $G$ are linearly independent, and $dim(C) = k = k_1 + k_2$.

%\smallskip 

\smallskip 
To show that $d' \geq d$, consider an arbitrary codeword $c \in C$ such that $c= r(gf_{11} , gf_{12} ,  \ldots , gf_{1l}) + t ( pgf_{22} ,  \ldots , pgf_{2l}) = (rgf_{11},g(rf_{12}+tf_{22}),\ldots, g(rf_{1l}+tf_{2l})$ for some polynomials $r, t \in \mathbb{F}[x]/\langle x^n-a\rangle$. 
Since $c$ is non-zero, at least one of the blocks of it must be non-zero. Since each block is a constacyclic code of minimum weight $d$, minimum weight of $c$, $d' \geq d$. 
\end{proof}
We have some important remarks about this theorem.

\begin{remark}

\begin{enumerate}
\item The lower bound on the minimum distance cannot be improved as there are cases where the lower bound is attained with equality. Consider, for example,  a 2-generator QT code $C$   as in Theorem 3.2 such that $p=1$, and $f_{1i}=f_{2i}$ for $i = 2,3,\dots, \ell$. Let $r=-t$, and consider the codeword 
$c= r(gf_{11} , gf_{12} ,  \ldots , gf_{1l}) + t ( 0,pgf_{22} ,  \ldots , pgf_{2l})$. Then $c = ( rf_{11}g,0,0,\dots,0)$. It follows that weight of $c$, $d$ is the weight of the constacyclic code generated by $rf_{11}g$, which is $d$.

%Thus, the lower bound for minimum distance of $C$ is $d'=d$, which is achieved when, for $i = 2,3,\ldots, l$, $f_{1i}=f_{2i}$

\item Even though the lower bound in the theorem cannot be improved, and it is much smaller than the 1-generator case, the actual minimum distance of $C$ is often much larger.

\item In the special case when $p=1$, we have  $k_1 = m - deg(g)$ and $k_2 = m - deg(1 \cdot g) = m - deg(g)$. Thus, $k_1 = k_2$, which means $dim(C) =$ $ k_1 + k_2 $ $= k_1 + k_1 = 2k_1$.

\end{enumerate}
\end{remark}

%\end{lemma}

\section{Search Methods}
We developed and implemented a search algorithm based on Theorem 3.2 as a well as a variation of it. We consider this method a generalization the ASR algorithm. Our search yielded many new QT codes some of them having additional desirable properties such as being  linear-complimentary dual(LCD), dual-containing, or reversible. Our search is set up the same way as the generalized ASR algorithm. Hence, we start with computing all CC codes of length $m$, partitioning them into equivalence classes using the algorithm in \cite{CCeq}, and picking one generator from  each equivalence class. We describe some special cases of the search below.

\newpage 
%the 2-generator QT form described in theorem 3.1 and also some of its variations described below to obtain many new QT codes with many of them being  as well. 

 %\hl{Original ASR vs generalized ASR, and how CC partition is relevant for the gen alg... -- I TALK ABOUT THIS EARLIER AFTER THEOREM 3.1 AND AM NOT SURE WHAT MORE TO ADD}

\addtolength{\hoffset}{2 cm}
\subsection{Case when $p=1$}
Let $g$ be the standard generator of a CC code $C_g$ of length $m$ and dim $k$, i.e. $x^m-a=gh$ with $k=\deg(h)=m-\deg(g)$. Let $C$ be a 2-generator QT code with a generator  of the form 
$$ \begin{bmatrix}
    gf_{11} & gf_{12} &  \cdots & gf_{1\ell} \\
    0 & gf_{22} &  \cdots & gf_{2\ell}
  \end{bmatrix} $$
where $deg(g) \geq 0$ and $\gcd(f_{ij},h)=1$.

%\hl{Why did you comment out the next subsection? Did you not find any good codes from that? The notes document indicates that you did.}

% \subsection{Using $g$ on the first block and $pg$ on the second block}
% Let $g$ be the standard generator of a CC code $C_g$ of length $m$ and dim $k$, i.e. $x^m-a=gh_1$ over a finite field $q$. Let $p$ be a polynomial over $q$ such that $p$ divides $h_1$. Let $h_2 = \frac{(x^m - a)}{p \cdot g}$. Let $C$ be a 2-gen QT code with a gen $G$ of the form 
% $$G =  \begin{bmatrix}
%     gf_{11} & gf_{12} &  \cdots & gf_{1\ell} \\
%     0 & pgf_{22} &  \cdots & pgf_{2\ell}
%   \end{bmatrix} $$

% where $\gcd(f_{1j},h_1)=1$ and $\gcd(f_{2j},h_2)=1$. 

\subsection{ $g_1=1$ and high degree $g_2$}

Let $g_1=1$ and let $g_2$ be the standard generator of a CC code $C_g$ of length $m$ and dim $k$, i.e. $x^m-a=g_2h$ with $k=\deg(h)=m-\deg(g_2)$. Let $C$ be a 2-generator QT code with a generator matrix $G$ of the form 
$$ \begin{bmatrix}
    f_{11} & f_{12} &  \cdots & f_{1\ell} \\
    0 & g_2f_{22} &  \cdots & g_2f_{2\ell}
  \end{bmatrix} $$
where $\deg(g_2)$ starts from the highest possible value and decreases subsequently and $\gcd(f_{2j},h)=1$.

\smallskip

\noindent The following method is not directly related to Theorem 3.2 but we were  able to obtain some good 2-generator QT codes from it.  

\subsection{The second block is a CC shift of the generators used in the first block}
Let $g$ be the standard gen of a CC code $C_g$ of length $m$ and dim $k$, i.e. $x^m-a=gh$ with $k=\deg(h)=m-\deg(g)$. Let $C$ be a 2-gen QT code with a gen $G$ of the form 
$$ \begin{bmatrix}
    gf_{11} & gf_{12} &  \cdots & gf_{1(\ell-1)} & gf_{1\ell} \\
    xgf_{1\ell} & gf_{11} & gf_{12} & \cdots & gf_{1(\ell-1)}
  \end{bmatrix} $$
where $deg(g) \leq 2$ and $\gcd(f_{ij},h)=1$.

\section{The New Codes}

Implementing the above search methods, we have obtained many new QT codes over GF(2), GF(3), GF(4), and GF(5) that have the following characteristics:

\begin{enumerate}
 
     \item All of these codes are new among the class of QT codes according to the database \cite{qcdatabase}.
     
    \item All of these codes have the the same parameters as BKLCs in \cite{database}.
   
    \item In many cases, the BKLCs in the database (\cite{database}) have indirect, multi-step constructions so it is more efficient and desirable to obtain them in the form of QT codes instead. Many of our codes possess this characteristic.
    
    \item A number of our codes have additional desirable properties such as  being dual-containing and linear complementary dual (LCD).
\end{enumerate}

 The following tables list the parameters, properties, and generators of these new QT codes. For clarification, $g_1$ represents the generator polynomial for the first block and $g_2$ represents the generator polynomial for the second block of the 2-generator QT code. All polynomials are listed by their coefficients  for a compact representation. For example, consider the $[39,24,6]_2$ code in Table 1 below whose $f_{11}$ is $x^{11} + x^{10} + x^8 + x^7 + x^6 + x^4 + x$. Its coefficients are $010010111011$ in increasing powers of $x$ from left to right. The block length of a CC code used to generate the QT code,  can be inferred from the data. For example, consider the $[39,24,6]_2$ code below. Observe that we have $[f_{11}]$,$[f_{12}]$, and $[f_{13}]$ in the last column for $[[f_{11}], \dots, [f_{1l}]]$, which means the index ($\ell$) is 3. This implies that the length of the CC code used is $m={n\over \ell}=\frac{39}{3}=13$. Also, in the tables below, we  indicate the shift constant $a$ of the field only if it is not equal to 1. 
 %\hl{The value of $a$ is never given in any of the tables below. Does this mean that they are all QC codes? If so, we can say this and delete $a$. -- UPDATED}

\addtolength{\hoffset}{-2.2 cm}

% insert image here

\makeatletter
\def\old@comma{,}
\catcode`\,=13
\def,{%
  \ifmmode%
    \old@comma\discretionary{}{}{}%
  \else%
    \old@comma%
  \fi%
}
\makeatother

%\hl{There is no $p$ in any of the tables. Does this mean that you never got any additional property from the new QT codes that contained $p$? You got the greatest number of new codes from this method, right? but no additional properties? And never found a reversible code?}

%\hl{In Thm 3.2, $g_1=g_2$ any way but there is no $p$ here. Don't we need to say $p=1$ in Table 1? }
 
 \begin{table}[htb]
 {
    \centering
    \caption{New QT Codes with $g=g_1=g_2$ that are LCD}

    \begin{tabular}{l|l|p{0.70\linewidth}}
    
$[n,k,d]_q$ & $a,g$ & $[[f_{11}],...,[f_{1\ell}]] , [[f_{21}],...,[f_{2\ell}]]$\\
\hline&&\\
$[39,24,6]_2 $ & $[11]$ & $[[010010111011],[011000110001],[111011011011]
] , [ [0],[010101110011],[0001011001]]$ \\ 
$[51,32,8]_2 $ & $[11]$ & $[[10000101111101],[0111010011000001],[0111000
011001101]],[[0],[1101101011010011],[1110001100110101]]$ \\
$[87,56,10]_2 $ & $[11]$ & $[[0011000001100000010110111001],[11110111110
10011011000101101],[000111000101100111111011111]],[[0],[1110001001110101000101
0011],[000011010100100110001011001]]$ \\
$[105,40,22]_2 $ & $[11]$ & $[[01010101001100101011],[010011010011101001
01],[11110100001001011],[0000110010010101111],[11100010000000100101]],[[0],[00
100101110101111001],[1101100001100001],[00110001110101001011],[1000010001101110
1]]$ \\ 
$[57,36,8]_2 $ & $[11]$ & $[[000101011100111001],[111010000110111011],[1
1000010101110101]],[[0],[010111011010011001],[101000100101111111]]$ \\
$[95,36,20]_2 $ & $[11]$ & $[[01110001000001111],[010100101011100001],[0
10111110011000111],[000110110001101011],[000000000000100001]],[[0],[1010101011
010111],[00001010110111101],[101000000111000101],[10100111110011]]$ \\
$[93,60,10]_2 $ & $[11]$ & $[[000010101010010010011110011001],[000001100
100000000100010111001],[110001110110011111100101]],[[0],[000101000101010010010
10001111],[11111111100010100101001011]]$ \\ 
$[63,40,8]_2 $ & $[11]$ & $[[1001010101011001001],[00001011110001010101]
,[01111001100000111101]],[[0],[01111101101010011001],[00100110100011110001]]$ \\
$[36,24,6]_3 $ & $2,[1]$ & $[[001021111101],[200120211122],[02102100121]],[
[0],[012010101221],[110200002221]]$ \\
$[54,36,8]_3 $ & $2,[1]$ & $[[112122022022122001],[022211020011212211],[212
202102022020122]],[[0],[200220211010000012],[10002021110222122]]$ \\
$[42,28,7]_3 $ & $2,[1]$ & $[[20110220012],[11102101212001],[22201012101002
]],[[0],[2120111110202],[22022120000212]]$

    \end{tabular}
    \label{tab:my_label}
}

 \end{table}
%\input{Tables/Sorth}
% \input{Tables/SorthRev}
%\hl{To be consistent with case 4.2 above, should we not call it $g_2$ in Table2 and Table 3?}
In  Table 2 and Table 3 below, we have codes with generators of the form $g_1=1$ and $g_2=pg_1=p$. So we simply give $g_2.$ 
%Also, for the $\mathbb{F}_4$ code in Table 3, $a^2+a+1=0$ and $b=a^2$. 
\begin{table}[H]
    \centering
    \caption{New QT Codes with $g_1=1$ that are LCD}
    \begin{tabular}{l|l|p{0.70\linewidth}}
$[n,k,d]_q$ & $a,g_2$ & $[[f_{11}],...,[f_{1\ell}]] , [[f_{21}],...,[f_{2\ell}]]$\\
\hline&&\\
$ [12,8,4]_5 $ & $2,[42411]$ & $[[103212],[201242]],[[0],[32]]$ \\
$ [14,8,5]_5 $ & $[1111111]$ & $[[0213422],[104033]],[[0],[1]]$ \\
$ [30,16,9]_5 $ & $2,[40101]$ & $[[142420004],[0234333133],[4302024312]],[[
0],[013143],[432312]]$ \\
$ [6,3,4]_5 $ & $[11]$ & $[[4],[21],[21]],[[0],[1],[4]]$ \\
$ [18,10,6]_5 $ & $2,[211]$ & $[[343214],[320101],[121443]],[[0],[1023],[30
14]]$ \\ 
$ [45,17,17]_5 $ & $[41]$ & $[[411224401],[143434113],[1103241],[432004303
],[202004042]],[[0],[24324023],[041344],[000422],[21112231]]$ \\
$ [6,4,3]_5 $ & $[111]$ & $[[422],[321]],[[0],[4]]$ \\
$ [18,8,8]_5 $ & $2,[42411]$ & $[[322144],[24241],[000442]],[[0],[31],[11]]$\\
$ [33,12,14]_5 $ & $[11111111111]$ & $[[23232410113],[0340212001],[02404034121]],[[0],[4],[4]]$\\
$ [38,20,11]_5 $ & $[1111111111111111111]$ &
$[[43320011101304443],[1300410301141140302]],[[0],[1]]$ \\
$ [ 12, 8, 4 ]_5 $ & $2,[42411]$ & $[[13333],[24401]],[[0],[04]]$ \\
$ [ 21, 8, 10 ]_5 $ & $[1111111]$ & $[[2244002],[2443132],[3033314]],[[0]
,[1],[1]]$\\
% $ [ 18, 8, 8 ]_5 $ & $2,[42411]$ & $[[144434],[322001],[414]],[[0],[02],[34
% ]]$\\
$ [ 12, 3, 8 ]_5 $ & $[11]$ & $[[01],[02],[1],[31],[34],[34]],[[0],[4],[4
],[3],[2],[3]]$ \\ 
$ [ 27, 10, 12 ]_5 $ & $[111111111]$ &
$[[34310241],[23310221],[434403341]],[[0],[3],[1]]$ \\
$ [ 30, 8, 16 ]_5 $ & $2,[42411]$ & $[[133114],[432411],[243314],[314243],[3
2331]],[[0],[12],[34],[31],[43]]$ \\
$ [ 18, 4, 12 ]_5 $ & $[111]$ & $[[232],[03],[342],[301],[031],[4]],[[0],
[3],[2],[1],[2],[1]]$ \\
$ [ 14, 3, 10 ]_5 $ & $[11]$ & $[[04],[24],[13],[1],[2],[04],[43]],[[0],[
4],[2],[1],[4],[2],[4]]$ \\
$ [ 28, 8, 15 ]_5 $ & $[1111111]$ & $[[1213104],[2141001],[1221304],[00122
1]],[[0],[2],[3],[4]]$ \\
$ [ 18, 3, 13 ]_5 $ & $[11]$ & $[[4],[42],[24],[04],[12],[3],[12],[31],[01
]],[[0],[4],[1],[4],[1],[4],[2],[2],[2]]$ \\
$ [ 34, 18, 10 ]_5 $ & $[11111111111111111]$ &
$[[0212240132442241],[20133301213221021]],[[0],[1]]$ \\
$ [ 42, 16, 16 ]_5 $ & $2,[4030102040301]$ &
$[[4232404223213],[30433103243322],[34024230010143]],[[0],[02],[31]]$ \\

$ [6,4,3]_7 $ & $[111]$ & $[[332],[456]],[[0],[6]]$ \\
$ [9,4,5]_7 $ & $[111]$ & $[[332],[456],[613]],[[0],[6],[3]]$ \\
$ [6,3,4]_7 $ & $[11]$ & $[[42],[05],[15]],[[0],[6],[2]]$ \\
$ [10,6,4]_7 $ & $[11111]$ & $[[10642],[20425]],[[0],[3]]$ \\
$ [12,6,6]_7 $ & $6,[141]$ & $[[6265],[6536],[6441]],[[0],[3],[32]]$ \\
$ [10,3,7]_7 $ & $[11]$ & $[[04],[13],[06],[04],[23]],[[0],[1],[1],[6],[5]]$ \\
$ [15,6,8]_7 $ & $[11111]$ & $[[11545],[332],[30203]],[[0],[4],[3]]$ \\
$ [16,6,9]_7 $ & $6,[141]$ & $[[1065],[0533],[1526],[0556]],[[0],[51],[41],[14]]$ \\
$ [15,4,10]_7 $ & $[111]$ & $[[233],[033],[434],[22],[246]],[[0],[1],[6],[1],[6]]$ \\
% $ [30,12,13]_7 $ & $[106010601]$ & $[[6502055016],[4062542362],[2620663161]],[[0],[65],[23]]$ \\
$ [16,3,12]_7 $ & $[11]$ & $[[56],[14],[54],[12],[53],[13],[6],[5]],[[0],[1],[3],[5],[2],[3],[5],[2]]$ \\
$ [18,3,14]_7 $ & $[11]$ & $[[05],[4],[63],[53],[1],[02],[14],[31],[64]],[[0],[1],[2],[5],[6],[6],[3],[5],[3]]$ \\
$ [21,4,15]_7 $ & $[111]$ & $[[033],[354],[303],[654],[631],[05],[265]],[[0],[2],[2],[3],[3],[4],[4]]$ \\
$ [8,3,6]_7 $ & $[11]$ & $[[35],[56],[03],[5]],[[0],[4],[2],[4]]$ \\
$ [26,14,9]_7 $ & $[1111111111111]$ & $[[260255156432],[5063322000532]],[[0],[1]]$ \\
$ [27,10,13]_7 $ & $[111111111]$ & $[[223013033],[543543144],[202056345]],[[0],[2],[6]]$ \\
$ [22,12,8]_7 $ & $[11111111111]$ & $[[4253164303],[10043110262]],[[0],[4]]$ \\
$ [33,12,15]_7 $ & $[11111111111]$ & $[[4253164303],[10043110262],[56200426441]],[[0],[4],[6]]$ \\

    \end{tabular}
    \label{tab:my_label}
\end{table}
 %\end{comment}

\begin{table}[H]
    \centering
    % \caption{New QT Codes that are Self-Orthogonal}
    \begin{tabular}{l|l|p{0.70\linewidth}}
$[n,k,d]_q$ & $a,g_2$ & $[[f_{11}],...,[f_{1\ell}]] , [[f_{21}],...,[f_{2\ell}]]$\\
\hline&&\\

$ [24,3,19]_7 $ & $[11]$ & $[[65],[36],[65],[3],[46],[1],[45],[62],[64],[63],[54],[63]],[[0],[5],[2],[2],[3],[5],[5],[2],[2],[2],[1],[3]]$ \\  

$ [30,3,24]_7 $ & $[11]$ & $[[65],[36],[65],[3],[46],[1],[45],[62],[64],[63],[54],[63],[01],[05],[46]],[[0],[5],[2],[2],[3],[5],[5],[2],[2],[2],[1],[3],[1],[6],[4]]$ \\
$ [34,18,11]_7 $ & $[11111111111111111]$ & $[[45313125015566331],[00225343656442513]],[[0],[2]]$ \\
$ [12,3,9]_7 $ & $[11]$ & $[[06],[54],[53],[26],[65],[21]],[[0],[5],[1],[5],[5],[5]]$ \\
$ [12,4,8]_7 $ & $[111]$ & $[[11],[332],[551],[441]],[[0],[1],[5],[1]]$ \\
$ [40,12,20]_7 $ & $6,[106010601]$ & $[[050501626],[1433612426],[2425560244],[2036634133]],[[0],[33],[24],[36]]$ \\
$ [45,10,26]_7 $ & $[111111111]$ & $[[06453443],[054056241],[611123654],[650341214],[103266104]],[[0],[6],[3],[1],[1]]$ \\
$ [38,3,31]_7 $ & $[11]$ & $[[3],[5],[56],[51],[65],[45],[64],[05],[2],[02],[46],[45],[05],[13],[56],[35],[12],[04],[63]],[[0],[6],[3],[5],[2],[3],[3],[3],[6],[5],[5],[5],[3],[5],[4],[3],[1],[6],[2]]$ \\

    \end{tabular}
    \label{tab:my_label}
\end{table}
 %\end{comment}

\begin{table}[H]
    \centering
    \caption{New QT Codes with $g_1=1$ that are Dual-Containing}
    \begin{tabular}{l|l|p{0.70\linewidth}}
    \\
$[n,k,d]_q$ & $a,g_2$ & $[[f_{11}],...,[f_{1\ell}]] , [[f_{21}],...,[f_{2\ell}]]$\\
\hline&&\\
$ [44,34,5]_3 $ & $2,[10202020001]$ & $[[1122210122122002202112],[2021220000
12101200012]],[[0],[112201210001]]$ \\
$ [20,16,3]_3 $ & $2,[11021]$ & $[[1102020111],[0020012212]],[[0],[201122]]
$ \\
$ [112,96,6]_3 $ & $2,[10201000000010101]$ & 
$[[20222110221200110002211011200001001221000211001101122111],[0000010210120110
0101210212022012101222110102211202012002]],[[0],[12121211002010202000202002001
1121001101]]$ \\
$ [136,118,6]_3 $ & $2,[2100120102111121121]$ & 
$[[00010212010010111122100101010121110201012110211202022100110012012201],[2020
2110122021012211012111022101000220100001210120202202022001121212]],[[0],[12221
220002011020102101111010222210120212120021221]]$ \\
$ [14,11,3]_4 $ & $[1],[4321]$ & $[[11aab1b],[0aa0ba]],[[0],[101b]]$ \\ 
$ [8,6,3]_7 $ & $6,[141]$ & $[[0654],[615]],[[0],[33]]$
    \end{tabular}
    \label{tab:my_label}
\end{table}
 \noindent For the $\mathbb{F}_4$ code in Table 3, $a^2+a+1=0$ and $b=a^2$.
 %\end{comment}

\addtolength{\hoffset}{-2.2 cm}

\newpage
\addtolength{\hoffset}{2.2 cm}

In addition to the additional properties that they posses, many of our new QT codes are  better than the BKLCs currently listed in the database \cite{database} for the reason that their constructions are far simpler. A QT code is more desirable than an arbitrary linear code for many reasons. It has a well understood  algebraic structure and its generator matrix is determined solely by its two rows. In comparison, the BKLCs in \cite{database} with the same parameters often have complicated constructions. For example, we found a $[30,16,9]_5$ code that has the same parameters as the comparable BKLC, as well as being LCD. Since it is a QT code, it has a single step construction. The code in the database \cite{database} on the other hand has a 6-step  construction to achieve the same parameters and lacks any additional properties. So, in many cases the QC  codes we have found have better structures with more desirable properties and with the same parameters as BKLCs. 
 
The codes that have additional properties are given in the tables above. The new codes that are not listed in the tables above are enumerated below. To save space, we do not write down their generators, which are readily available from the authors. Moreover, these codes have been reported to the owner of the database \cite{qcdatabase}. Once they are added to the database, their generators could also be available.
\vspace{-0.5cm}

\begin{multicols}{4}
{\renewcommand\labelitemi{}
\begin{itemize}[leftmargin=*]
\item $[90,65,8]_2 $
\item $[45,30,6]_2 $
\item $[102,68,10]_2 $
\item $[75,50,8]_2 $
\item $[78,52,8]_2 $
\item $[48,25,10]_2 $
\item $[111,38,24]_2 $
\item $[92,47,14]_2 $
\item $[84,43,14]_2 $
\item $[54,34,8]_2 $
% \item $[39,24,6]_2 $
% \item $[51,32,8]_2 $
% \item $[105,40,22]_2 $
% \item $[57,36,8]_2 $
% \item $[95,36,20]_2 $
% \item $[93,60,10]_2 $
\item $[45,28,8]_2 $
% \item $[63,40,8]_2 $
\item $[90,58,10]_2 $
\item $[30,23,4]_3 $
\item $[26,22,3]_3 $
% \item $[44,34,5]_3 $
\item $[33,22,6]_3 $
\item $[45,30,7]_3 $
\item $[12,4,7]_4 $
\item $[26,20,4]_4 $
\item $[18,14,3]_4 $
\item $[16,6,8]_4 $
\item $[21,11,7]_4 $
% \item $[12,8,4]_5 $
\item $[10,7,3]_5 $
% \item $[14,8,5]_5 $
\item $[16,6,8]_5 $
\item $[10,6,4]_5 $
\item $[8,3,5]_5 $
\item $[36,14,14]_5 $
\item $[45,10,24]_5 $
\item $[15,6,8]_5 $
\item $[45,16,18]_5 $
\item $[39,14,16]_5 $
% \item $[30,16,9]_5 $
\item $[9,4,5]_5 $
\item $[55,17,23]_5 $
% \item $[6,3,4]_5 $
% \item $[18,10,6]_5 $
% \item $[45,17,17]_5 $
\item $[33,17,10]_5 $
% \item $[6,4,3]_5 $
% \item $[12,8,4]_5 $
% \item $[18,10,6]_5 $
\item $[38,28,6]_5 $
\item $[44,24,11]_5 $
\item $[20,6,12]_7 $ 
\item $[16,10,5]_7 $ 
\item $[35,8,21]_7 $ 
\item $[14,8,6]_7 $ 
\item $[18,10,7]_7 $

\end{itemize}
}
\end{multicols}

\vspace{-0.5cm}

\section{New Record Breaking Binary Linear Codes}

Using the search method described in section 4.2, we were able to obtain a record-breaking 2-generator binary QT code with a minimum distance higher than the BKLC of the same length and dimension over $GF(2)$ given in \cite{database}.  The generators of this $[111,38,25]_2$-code (in the form of Theorem 3.2)  are\\ 
{\footnotesize{$g_1=1$} \\ 
$g_2=x^{36} + x^{35} + x^{34} + x^{33} + x^{32} + x^{31} + x^{30} + x^{29} + x^{28} + x^{27} + x^{26} + x^{25} + x^{24} + x^{23} + x^{22} + x^{21} + x^{20} + x^{19} + x^{18} + x^{17} + x^{16} + x^{15} + x^{14} + x^{13} + x^{12} + x^{11} + x^{10} + x^9 + x^8 + x^7 + x^6 + x^5 + x^4 + x^3 + x^2 + x + 1$ 
$f_{11} = x^{36} + x^{35} + x^{32} + x^{27} + x^{26} + x^{25} + x^{24} + x^{23} + x^{22} + x^{21} + x^{19} + x^{18} + x^{15} + x^{14} + x^{10} + x^9 + x^8 + x^7 + x^6 + x^5 + x^4 + x^3 + x^2$.\\ 
\noindent $f_{12} = x^{35} + x^{34} + x^{33} + x^{30} + x^{27} + x^{26} + x^{24} + x^{22} + x^{20} + x^{19} + x^{17} + x^{16} + x^{13} + x^{12} + x^{11} + x^{10} + x^9 + x^7 + x^6 + x^4 + x^3 + x^2 + x$.\\  
\noindent$f_{13} = x^{35} + x^{33} + x^{28} + x^{27} + x^{26} + x^{24} + x^{22} + x^{18} + x^{16} + x^9 + x^4 + x^3 + x^2 + x + 1$.\\ 
$f_{21}=0, \qquad f_{22}=1 \qquad f_{23}=1$} \\

Moreover, applying the standard method of extending this  code, we  obtain another record-breaking linear code with parameters $[112,38,26]_2$.

%\newpage


\begin{thebibliography}{}



\bibitem{database} M. Grassl, Code Tables: Bounds on the parameters of of codes, online, \url{http://www.codetables.de/}

\bibitem{ASR}N. Aydin, I. Siap, D. Ray-Chaudhuri, The structure of 1-generator quasi-twisted codes and new linear codes. Design Code Cryptogr.  24, 313-326 (2001)


\bibitem{GenASR}N. Aydin, J. Lambrinos, R. O. VandenBerg, On Equivalence of Cyclic Codes, Generalization of a Quasi-Twisted Search Algorithm, and New Linear Codes.  Design Code Cryptogr. 87, 2199-2212 (2019)

\bibitem{qt1}  N. Aydin and I. Siap, \textit{New quasi-cyclic codes over $\mathbb{F}_5$}, Appl. Math. Lett., 15, 833-836, 2002.


\bibitem{qt2} R. Daskalov and P. Hristov, New quasi-twisted degenerate ternary linear codes, IEEE Transactions on Information Theory, vol. 49, no. 9, pp. 2259-2263, Sept. 2003, doi: 10.1109/TIT.2003.815798.

\bibitem{qt3} R. Ackerman, and N. Aydin, New quinary linear codes from quasi-twisted codes and their duals, Applied Mathematics Letters, Vol 24, No 4, pp 512-515, April 2011.

\bibitem{qt4} R. Daskalov, and P. Hristov, \textit{Some new quasi-twisted ternary linear codes}, JACODESMATH, 2(3), 211-216, 2015.

\bibitem{2gen1} T. A. Gulliver and V. K. Bhargava, \textit{Two new rate 2/p binary quasi-cyclic codes}, IEEE Trans.  Inf. Theory, 40(5), 1667-1668, 1994.

\bibitem{2gen2} E. Z. Chen, \textit{An explicit construction of 2-generator quasi-twisted codes},  IEEE Trans. Inf. Theory, 54(12), 5770-5773, 2008.


\bibitem{qcdatabase} E. Z. Chen, Quasi-Cyclic Codes: Bounds on the parameters of of QC codes, online, \url{http://www.tec.hkr.se/~chen/research/codes/qc.htm}


\bibitem{magma} Magma computer algebra system, online, \url{http://magma.maths.usyd.edu.au/}



\bibitem{cycliceq}  N. Aydin,  R. O.  VandenBerg, A New Algorithm for Equivalence of Cyclic Codes and Its Applications,  arxiv preprint, \url{https://arxiv.org/abs/2107.00159}, 2021. 

\bibitem{CCeq} D. Akre, N.  Aydin,  M. J.   Harrington, S. Pandey,  A Generalization of Cyclic Code Equivalence Algorithm to Constacyclic Codes,  arxiv preprint, \url{https://arxiv.org/abs/2108.08619}, 2021. 


\bibitem{constx} DR. askalov, P. Hristov,  Some new ternary linear codes. JOURNAL OF ALGEBRA COMBINATORICS DISCRETE STRUCTURES AND APPLICATIONS. 4. 227-234 (2017)


\bibitem{Twistulant} N. Aydin, T. H.  Guidotti, P.  Liu,  A. Shaikh, and R. O. VandenBerg, Some generalizations of the ASR search algorithm for quasitwisted codes. Involve, a Journal of Mathematics 13, no. 1 (2020): 137-148.


%\bibitem{Twistulant} Some generalizations of the ASR search algorithm for quasitwisted codes Nuh Aydin, Thomas H. Guidotti, Peihan Liu, Armiya S. Shaikh and Robert O. VandenBerg (2020)

%\bibitem{mt}Aydin N., Halilovic, A.: A generalization of quasi-twisted codes: Multi-twisted codes. Finite Fields Appl.  45, 96-106 (2017)

\bibitem{gf7} N. Aydin, N. Connolly, M.  Grassl, Some results on the structure of constacyclic codes and new linear codes over GF(7) from quasi-twisted codes. Adv. Math. Commun. 11, 245-258 (2017)



\bibitem{NPhard} A. Vardy, The intractability of computing the minimum distance of a code. IEEE Trans. Inform. Theory. 43 1757-1766 (1997)


 
%\bibitem{Quantumoriginal1}Steane, A. M.: Error correcting codes in quantum theory. Phys. Rev. Lett. 77, 793-797 (1996)

%\bibitem{Quantumoriginal2}Calderbank, A. R., Shor, P. W.: Good quantum error-correcting codes exist. Phys. Rev. A. 54, 1098-1106 (1996)


%\bibitem{oto}Parra-Avila, B., Permouth, S., Szabo, S.: Dual Generalizations of the Concept of Cyclicity of Codes. Adv. Math. Commun. 3, 227-234 (2009)

%\bibitem{pc2011}Matsuoka, M.: $\theta$-Polycyclic codes and $\theta$-sequential codes over finite field. Int. J.  Algebra. 5, 65-70 (2011)

%\bibitem{pc2016}Alahmadi, A., Dougherty S., Leroy, A., Solé, P.: On the duality and the direction of polycyclic codes. Adv. Math. Commun. 10, 921-929 (2016)

%\bibitem{pc2020}Fotue-Tabue, A., Martinez-Moro, E., Blackford J. T.: On polycyclic codes over a finite chain ring, Adv. Math. Commun. 14, 455-466 (2020)

 
%\bibitem{Lid}Rudolf, L., Harald, N.: Introduction to finite fields and their applications. Cambridge University Press (1986)

%\bibitem{dna} Abualrub, T., Ghrayebb, A., Zeng, X.: Construction of cyclic codes over GF(4) for DNA computing. J. Franklin Inst. 343, 448-457 (2006)

%\bibitem{dna2}Oztas E. S., Yildiz, B., Siap, I.: A novel approach for constructing reversible codes and applications to DNA codes over the ring $F_{2}[u]/(u^{2k}-1)$. Finite Fields Appl. 46, 217-234 (2017)

%[17]
%\bibitem{src18}Gao, J., Gao, Y., Fu, F.: Quantum codes from cyclic codes over the ring $\mathbb{F}_{q}+v_{1}\mathbb{F}_{q}+...+v_{r}\mathbb{F}_{q}$. Appl. Algebra Engrg. Comm. Comput. 30, 161-174 (2019)

%[18]
%\bibitem{GF17&19}Gulliver, T., Venakaiah, V.: Construction of quasi-twisted codes and enumeration of defining polynomials. J. Algebra Comb. Discrete Struct. Appl. 7, 1-18 (2019)

%[19]
%\bibitem{src6}Bag, T., Dinh, H. Q., Upadhyay, A. K., Yamaka, W.: New non-binary quantum codes from cyclic codes over product rings. IEEE Communications Letters. 24, 486-490 (2019)

%[20]
%\bibitem{src32} Qian, J., Zhang, L.: Nonbinary quantum codes derived from repeated-root cyclic codes, Modern Phys. Lett. B. 27, 1350053 (2013)

%\bibitem{seq} Hou, X., Lopez-Permouth, S., Parra-Avila, B.: Rational power series, sequential codes and periodicity of sequences. J. Pure Appl. Algebra. 213, 1157-1169 (2009)

%\bibitem{1972} Peterson, W. W., Weldon, E. J.: Error Correcting Codes. MIT Press (1972)

%[23]
%\bibitem{src14} Bag, T., Bandi, R., Chinnakum, W., Dinh, H., Upadhyay, A.: On the structure of cyclic codes over $F_{q}RS$ and applications in quantum and LCD codes constructions. IEEE Access. 8, 18902-18914 (2020)

%[24]
%\bibitem{src30}Özen, M., Özzaim, T., İnce, H.: Skew quasi cyclic codes over $\mathbb{F}_{q}+v\mathbb{F}_{q}$. J. Algebra Appl. 18, 1950077 (2018)

%[25]
%\bibitem{src29} Fu, F., Gao, J., Ma F.: New non-binary quantum codes from constacyclic codes over $\mathbb{F}_q[u,v]/\langle u^{2}-1, v^{2}-v, uv-vu\rangle$. Adv. Math. Commun. 13, 421-434 (2019)

%[26]
%\bibitem{src5}Ashraf, M., Bag, T., Mohammad, G., Upadhyay, A.: Quantum codes from cyclic codes over the ring $\mathbb{F}_p[u] / \langle u^3-u \rangle$. Asian-Eur. J. Math. 12, 2050008 (2020)

%[27]
%\bibitem{src26}Koroglu, M., Siap, I.: Quantum Codes From A Class of Constacyclic Codes over Group Algebras. Malays. J. of Math. Sci. 11, 289-301 (2017)

%[28]
%\bibitem{src19} Fu, F., Gao, J., Ma, F.: Constacyclic codes over the ring ${\mathbb {F}}_q+v{\mathbb {F}}_q+v^{2}{\mathbb {F}}_q$ and their applications of constructing new non-binary quantum codes. Quantum Inf. Process. 17, 122 (2018)

%[29]
%\bibitem{src13} Bag, T., Dinh, H., Upadhyay, A., Ashraf, M., Mohammad, G., Chinnakum, W.: New quantum codes from a class of constacyclic codes over Finite Commutative Rings. J. Algebra Appl. 19, 2150003 (2020)



% \bibitem{smallkeys} Heyse, Stefan, von Maurich, Ingo
% and G{\"u}neysu, Tim", Smaller Keys for Code-Based Cryptography: -MDPC McEliece Implementations on Embedded Devices, in "Cryptographic Hardware and Embedded Systems - CHES 2013", editors"Bertoni, Guido and Coron, Jean-S{\'e}bastien", Springer Berlin Heidelberg, Berlin, Heidelberg, 2013, pp273--292.






\end{thebibliography}
\end{document}